\documentclass{amsart}                     
\usepackage{graphicx}
\usepackage{amsmath}
\usepackage{amssymb}
\usepackage[english]{babel}
\usepackage{relsize}
\usepackage{amsthm}
\usepackage{mathrsfs}
\usepackage[latin1]{inputenc}
\usepackage{subfigure}
\usepackage{color}
\usepackage{comment}
\usepackage{tikz}
\usepackage{marginnote}
\DeclareMathOperator{\Tr}{Tr}

\newcommand{\TT}{\mathbb{T}}
\newcommand{\CC}{\mathbb{C}}
\newcommand{\RR}{\mathbb{R}}
\newcommand{\ZZ}{\mathbb{Z}}
\newcommand{\Aa}{\mathbb{A}}
\newcommand{\vv}{\mathrm{v}}
\newcommand{\norm}[1]{\|{#1}\|}

\def\dist{\mathrm{dist}}
\def\Anal{{\mathscr{A}}}
\def\Emb{{\mathscr{E}}}
\newcommand{\abs}[1]{|{#1}|}

\def\E{\varepsilon}
\def\bt{\bar{T}}
\def\cE{\mathfrak{e}}
\def\cN{\mathfrak{N}}
\newtheorem{problem}{Problem}
\newtheorem{remark}{Remark}
\newtheorem{theorem}{Theorem}
\newtheorem{lemma}{Lemma}

\begin{document}

\title[Orbit determination in regular zones: confidence region]{Orbit determination for standard-like maps: asymptotic expansion of the confidence region in regular zones}
  \thanks{This work was supported by the National Group of Mathematical Physics (GNFM-INdAM) through the project ``Orbit Determination: from order to chaos'' (Progetto Giovani 2019).}


\author{Stefano Mar\`o}


\address{Dipartimento di Matematica, Universit\`a di Pisa, 
  Largo B. Pontecorvo, 5, Pisa, Italy}

              \email{stefano.maro@unipi.it}           


\begin{abstract}
  We deal with the orbit determination problem for a class of maps of
  the cylinder generalizing the Chirikov standard map. The problem
  consists of determining the initial conditions and other parameters
  of an orbit from some observations. A solution to this problem goes
  back to Gauss and leads to the least squares method. Since the
  observations admit errors, the solution comes with a confidence
  region describing the uncertainty of the solution itself.  We study
  the behavior of the confidence region in the case of a simultaneous
  increase of the number of observations and the time span over which
  they are performed. More precisely, we describe the geometry of the
  confidence region for solutions in regular zones. We prove an estimate of the trend of the
  uncertainties in a set of positive measure of the phase space, made
  of invariant curve. Our result gives an analytical proof of some
  known numerical evidences.

\end{abstract}

\maketitle

\section{Introduction}
\label{intro}
Orbit determination is a classical problem in applied Celestial
Mechanics. It consists of recovering information on some parameters
(initial conditions or dynamical parameters) of a model, given some
observations.  The first notable result was obtained by Gauss in the
XIX century \cite{gauss}. He was able to recover the orbit of Ceres
given the $21$ observations made by Piazzi in different nights. Gauss
method was composed of two steps. First, an approximation of the
solution was computed, then, the least squares method was applied to
improve the first approximation.

This strategy is still in use nowadays and the applications have
become wide-ranging. The accurate determination of orbits of NEOs is
essential in the impact monitoring activity. On the other hand, the
targets of many space missions include the determination of some
unknown parameter. Typical examples are the ESA/JAXA BepiColombo
mission to Mercury, the NASA JUNO and ESA JUICE missions to Jupiter.

The result of an orbit determination process (called nominal solution)
always comes with a confidence region, summarizing the uncertainty of
the result itself. Its behavior is of crucial importance in applied
problems, for examples, it is at the base of the definition of the
impact probability in impact monitoring \cite{mv}. Hence, it is
important to study the confidence region as the number of observations
grows. However, due to the nature of the observation process, an
increase of the number of observations comes with a simultaneous
increase of the time interval over which they are performed.

The study of the practical feasibility of orbit determination in case
of a simultaneous increase of the number of observations and the time
span over which they are performed have been studied numerically in
\cite{ssm,smil}. The authors considered as a model the Chirikov
standard map \cite{chir}. This map is used to approximate many
phenomena in Celestial Mechanics (see \cite{cell}) and shows both
regular and chaotic dynamics.  The authors constructed the
observations by adding some noise to a true orbit of the map. Then,
they set up an orbit determination process to recover the true orbit.
The experiments showed that the result crucially depends on the
dynamics. More precisely, if the observations come from a chaotic
orbit, then the orbit determination process has to face the problem of
the so-called computability horizon. This prevents orbit determination
from being performed if the time span of the observations is too large.
Moreover, at least until the computability horizon, the uncertainties
decrease exponentially (w.r.t the number of observations). On the
other hand, if the observations come from a regular orbit (on an
invariant curve), then the problem of the computability horizon is no
more present and the uncertainties decrease polynomially. Moreover, if
together with the initial conditions a dynamical parameter has to be
determined, then in both the regular and chaotic case the
uncertainties decrease polynomially.

In this paper we give an analytical proof of some of the just
described numerical results. We will deal only with the regular case, in
which only the initial conditions have to be estimated. We consider a
class of perturbations of the integrable twist map, to which KAM theory
applies and generalizes the Chirikov standard map (see \cite{sm}).  We
will prove that there exists a set $\mathcal{K}$ of positive measure
of nominal solutions whose uncertainties satisfy the numerical
estimates in \cite{ssm,smil}. More precisely, the set $\mathcal{K}$
consists of invariant curves on which the dynamics is conjugate to a
Diophantine rotation. Hence, we will describe the behavior of the
confidence region for nominal solutions on such invariant curves as
the number of observations and the corresponding time span grow
simultaneously. Here, the Diophantine condition will play a crucial
role.

The paper is organized as follows. In Section \ref{statement} we
describe the orbit determination process that allows to give a formal
and general statement of the problem. Moreover, we will present our
main result, that will be proved in the subsequent sections. Section \ref{outline} is dedicated to give an outline of the proof of our main result. The purpose of this section is to give an informal idea of the proof and guide the reader through the detailed proof, presented in Section \ref{rig_proof}. The proof is divided in two steps: first we describe the dynamics of our class
of maps and introduce the set $\mathcal{K}$ in subsection \ref{sec:KAM}; subsequently, subsection
\ref{sec:confidecne} is dedicated to the study of the confidence
region for nominal solutions on invariant curves and to give the proof of our main result. In Section \ref{sec:comparison} we interpret our theoretical results in the light of the known numerical evidences. Section \ref{sec:conclusions} is dedicated to some
conclusions and statements of future work.

\section{Statement of the problem and main result}\label{statement}

Let $\mathbb{A} = \TT\times\RR$ be the cylinder, where we denote $\TT =\RR/\ZZ$. Consider a diffeomorphism $S: \Sigma \rightarrow \mathbb{A} $ defined on a strip $\Sigma = \TT\times (a,b)$ for $a<b$. Given an initial condition $(x,y)\in\Sigma $ and an integer $n$, we denote the $n$-th iterate by $(x_n,y_n) = S^n(x,y)$ and the corresponding orbit by $(x_n,y_n)_{n\in\ZZ}$ (note that the generic initial value for $n=0$ is denoted as $(x,y)$).
Beside the true orbits, suppose that we have been observing the evolution of the state of a system modeled by $S$ and we have got the observations $O_n=(\bar X_n,\bar Y_n)$ for $|n|\leq N$.
Following \cite{mg2010} we set up an orbit determination process in order to find the orbit of $S$ that better approximates, in the least squares sense, the given observations.
We first define the residuals
\begin{eqnarray*}
  &\xi_n(x,y)=O_n-S^n(x,y)=
  \left(
  \begin{array}{cc}
    \xi_n^x(x,y)\\
    \xi_n^y(x,y)
  \end{array}
  \right), \\
&\xi_n^x (x,y)= \bar X_n-x_n(x,y) \mod 2\pi,\quad \xi_n^y (x,y)= \bar Y_n-y_n(x,y).
\end{eqnarray*}

\noindent Subsequently, we say that the least squares solution $(x_0,y_0)$ is a minimizer (at least locally) of
the target function
$$
Q(x,y) = \frac{1}{2N+1}\sum_{|n|\leq N}\xi_n(x,y)^T\xi_n(x,y)=
\frac{1}{2N+1}\sum_{|n|\leq N}\left[(\xi_n^x)^2 + (\xi_n^y)^2\right].
$$

We will not be concerned with the existence and computation of the
minimum. This is a very delicate task, solved via iterative schemes
such as the Gauss-Newton algorithm and the differential
corrections. These algorithms crucially depend on the choice of the
initial conditions. See \cite{gbm}, \cite{mbbg} for some recent
results on this topic for the asteroid and space debris cases.  In the
following we will always suppose that the least squares solution
$(x_0,y_0)$ exists and we will refer to it as the nominal
  solution.

Since the observations contain errors, values of $(x,y)$ that make the target function a little bigger than the minimum $Q_0 = Q(x_0,y_0)$ are acceptable. This leads to the definition of the confidence region as
\[
\mathcal{Z} =\left\{ (x,y)\in\Aa \: : \: Q(x,y) \leq Q_0 +\frac{\sigma^2}{2N+1}   \right\},
\]
where $\sigma$ is chosen depending on statistical properties and bounds the acceptable errors; for our purposes, one can keep in mind $\sigma=1$. Expanding $Q(x,y)$ around the nominal solution (the minimum) $(x_0,y_0)$ up to second order we get
\begin{multline*}
Q(x,y) \sim Q(x_0,y_0)  + \\
\frac{1}{2N+1}
\left(\begin{array}{l}
  x-x_0 \\
  y-y_0
\end{array}
\right)^T
\sum_{|n|\leq N}\left[(DS^n)^TDS^n + (DS^n)\xi_n^T\right]_{(x_0,y_0)} \left(\begin{array}{l}
  x-x_0 \\
  y-y_0
\end{array}
\right)    .
\end{multline*}
Here, we denoted
\[
DS^n(x,y) =
\left(
\begin{array}{cc}
\frac{\partial x_n}{\partial x}(x,y) & \frac{\partial x_n}{\partial y}(x,y)  \\
\frac{\partial y_n}{\partial x}(x,y) & \frac{\partial x_n}{\partial y}(x,y)
\end{array}
\right).
\]
Under the hypothesis that the residuals corresponding to the nominal solution are small, we can neglect the term $\xi_n^T(x_0,y_0)$. Then, we define the normal matrix 

  \begin{equation}
\label{normmat}
    C_N(x,y) := \sum_{|n|\leq N}(DS^n)^T(x,y)DS^n(x,y)  
  \end{equation}
and the associated Covariance matrix
\[
\Gamma_N(x,y) = \left[C_N(x,y)\right]^{-1}. 
\]

Note that, $C_N(x,y)$ is positive definite since $S$ is a diffeomorphism and $DS^n$ has rank $2$.  
Hence, the confidence region can be approximated by the confidence ellipse
  
\[
\mathcal{E}_N(x_0,y_0) 
=\left\{ (x,y)\in\Aa \: : \:
\left(\begin{array}{l}
  x-x_0 \\
  y-y_0
\end{array}
\right)^T
C_N(x_0,y_0)
\left(
\begin{array}{l}
  x-x_0 \\
  y-y_0
\end{array}
\right)
\leq \sigma^2  \right\}.
\]
The Covariance matrix $\Gamma_N$ describes $\mathcal{E}_N$ in the sense that the axes are proportional to the square root of the eigenvalues of $\Gamma_N(x_0,y_0)$ and are directed along the corresponding eigenvectors.
The region $\mathcal{E}_N$ represents the uncertainty of the nominal solution: the values inside $\mathcal{E}_N$ are acceptable and the projections of $\mathcal{E}_N$ on the axes, denoted as $\sigma_x$ and $\sigma_y$ represent the (marginal) uncertainties of the coordinates. See Figure \ref{figure:ellipse}.
We remark that the matrices $C_N$, $\Gamma_N$ also have a probabilistic interpretation, see \cite{mg2010}.

\begin{figure}[h!]
   \centering
  \begin{tikzpicture}
    \draw [help lines, ->]  (-1.1,0) -- (9,0);
    \draw [help lines, ->] (0,-1.1) -- (0,6.5);
    \draw [rotate around={30:(5,3)}] (5,3) ellipse (100pt and 50pt);
    \draw [rotate around={30:(5,3)}] (1.5,3) -- (8.5,3);
    \draw [rotate around={30:(5,3)}] (5,1.2) -- (5,4.8);
    \draw [dotted] (1.82,2) -- (1.82,-0.5);
    \draw [dotted] (8.17,5) -- (8.17,-0.5);
    \draw [dotted] (-0.5,0.65) -- (4,0.65);
    \draw [dotted] (-0.5,5.35) -- (8,5.35);
    \draw [thick] (1.82,0) -- (8.17,0);
    \draw [thick] (0,0.65) -- (0,5.35);
    \node [below right] at (9,0) {$x$};
    \node [left] at (0,3) {$\sigma_y$};
    \node [below] at (5,0) {$\sigma_x$};
    \node  at (5.7,2.8) {$(x_0,y_0)$};
    \filldraw (5,3) circle[radius=2pt];
    \node [above left] at  (0,6.5) {$y$};
  \end{tikzpicture}
  \caption{The confidence ellipse for the nominal value $(x_0,y_0)$. The values $\sigma_x,\sigma_y$ represent the marginal uncertainties of $x_0,y_0$ respectively. }
   \label{figure:ellipse}
\end{figure}
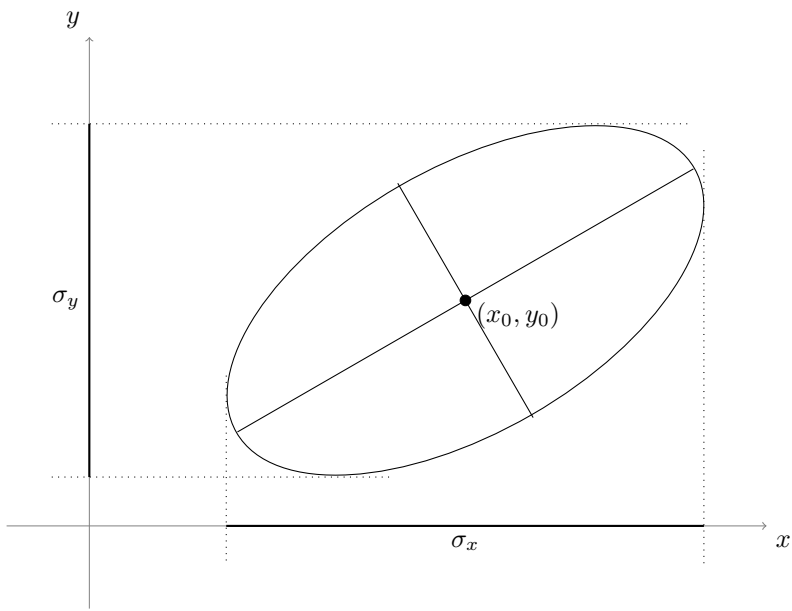

From the point of view of the applications, (e.g. impact monitoring \cite{mv}), it is of fundamental importance to know the shape and size of the confidence ellipse $\mathcal{E}$. Hence, the question that we are trying to address, stated in a broad sense, is the following:

\begin{problem}\label{p1}
  Given a diffeomorphism $S(x,y)$ of the cylinder and a nominal solution $(x_0,y_0)$, describe $\mathcal{E}=\mathcal{E}_N(x_0,y_0)$ for large $N$.
\end{problem}
\begin{remark}
The solution of the problem passes through the computation of the eigenvalues and eigenvectors of the matrix $ \Gamma_N(x_0,y_0)$ for large $N$. Note that they crucially depend on the dynamics, since we have to compute the linearized system along an orbit.
\end{remark}
\begin{remark}\label{remextend}
The statement of the problem can be generalized to different situations, such as more degrees of freedom and the case of continuous dynamics and flows. As an example, on can consider an Hamiltonian system with $p$ degrees of freedom defined by the Hamiltonian $H(\theta,I)=H_0(I)+\varepsilon H_1(\theta,I)$ and the corresponding flow $\Phi_t(\theta,I)$. Suppose to have the observations $(\bar{\theta}_n,\bar{I}_n)_{|n|\leq N}$ corresponding to $\Phi_n(\theta,I)$. As before we can set up an orbit determination process and define a $2p$-dimensional confidence ellipsoid and a corresponding Covariance Matrix. See \cite{mg2010}. 
  \end{remark}

\subsection{Main result}
In this paper we will give an answer the Problem \ref{p1} in a particular case.\\ 
Suppose that the map $S(x,y)$ is defined in the strip $\TT\times [a,b]$ with $b-a\geq 1$ and has the form
\begin{equation}
  \label{kammap}
\left\{
\begin{split}
  x_1 &= x+y+  f(x,y) \\
  y_1 &= y+  g(x,y).
\end{split}
\right.
\end{equation}
Here $f,g$ are bounded real analytic functions. 
Moreover, we suppose that the map is exact in the sense that there exists an analytic function $V:\TT\times (a,b) \rightarrow \RR$ such that
\[
y_1 dx_1 -ydx = dV(x,y).
\]
%
To state the result, we denote
\[
\E =\norm{f} +\norm{g},
\]
where, for the moment, $\norm{\cdot}$ represents a norm in the space of analytic functions that will be specified in Section \ref{secnot} (see \eqref{defe}).
\begin{remark}
As an example of map \eqref{kammap} we can consider the case $f(x,y)=g(x,y)=k\phi(x)$ with $k\in\RR$ and $\phi$ real analytic and $1$-periodic. The case $\phi(x)=\sin(2\pi x)$ represents the Chirikov standard map. Our main result, Theorem \ref{main}, will apply for small enough values of $k$. This kind of maps are prototypes for many applications in Celestial Mechanics. See \cite{sm}.
  \end{remark}

\noindent For $\E=0$, the map $S$ is linear and takes the form $S(x,y)=A\: (x,y)^T$ with
\[
A=
\left(
\begin{array}{cc}
  1 & 1 \\
  0 & 1
\end{array}
\right).
\]
This map is integrable and the phase space is foliated by invariant curves of the form $y = const$.
The normal matrix $C_N$ is independent on the nominal solution $(x_0,y_0)$ and reads
\[
C_N = \sum_{|n|\leq N}(A^n)^T A^n =
\left(
\begin{array}{cc}
  2N+1 & 0 \\
  0 & 2N+1+\sum_{|n|\leq N}n^2
\end{array}
\right).
\]
Hence, for large $N$,
\[
\Gamma_N = C^{-1}_N = 
\left(
\begin{array}{cc}
  \frac{1}{2N+1} & 0 \\
  0 & \frac{1}{2N+1+\sum_{|n|\leq N}n^2}
\end{array}
\right)
=
\left(
\begin{array}{cc}
  \frac{1}{2N+1} & 0 \\
  0 & \frac{3}{2N^3} +O(N^{-4})
\end{array}
\right).
\]

From this, we get the description of the confidence ellipse: the major axis is directed along the curve $y=y_0$ and has length with leading term $1/\sqrt{2N}$. The minor axis (orthogonal to the major axis) has length with leading term $\sqrt{3}/\sqrt{2N^3}$. Hence, the two coordinates of the nominal solution $(x_0,y_0)$ have different uncertainties: $\sigma_x\sim 1/ \sqrt{N}$ while $\sigma_y\sim 1/N^{3/2}$ for large values of $N$.

We are going to show how the linear situation is modified when one considers the map $S$ for $\E\neq 0$.\\
To state the theorem, we will denote
$\lambda_+=\lambda_+(x,y),\lambda_-=\lambda_-(x,y)$ the eigenvalues of $\Gamma_N(x,y)$ and $u_+=u_+(x,y),u_-=u_-(x,y)$ the corresponding eigenvectors. The proof will be given in Section \ref{sec:confidecne}.


%

\begin{theorem}\label{main}
Let $S$ be a diffeomorphism as before. There exist two positive constants $\underline{\kappa},\overline{\kappa}$ such that, for every $\E<\underline{\kappa}$ there exists a set $\mathcal{K}_\E$ of positive measure such that if
\begin{equation*}
 (x,y)\in \mathcal{K}_\E,
\end{equation*}
then for every $N> \frac{1}{\underline{\kappa}}$
\begin{align*}
 \lambda_+ &= \frac{1}{2N+1}(1 + \cE_+),
 \quad   &
u_+&=
\left(
\begin{array}{l}
  1 \\
  0 
\end{array}
\right)
+
\cE_u\left(
\begin{array}{l}
 1  \\
 1
\end{array}
\right),
 \\
 \lambda_- &= \frac{1}{2N+1+\sum_{|n|\leq N}n^2}(1 + \cE_-),   &
 u_-&=
\left(
\begin{array}{l}
  0 \\
  1 
\end{array}
\right)
+
\cE_u\left(
\begin{array}{l}
1 \\
  1
\end{array}
\right),
 \end{align*}
 where $\cE_{\pm}=\cE_{\pm}(x,y)$, $\cE_{u}=\cE_{u}(x,y)$ are real analytic functions with $|\cE_{\pm}|,|\cE_{u}|<\overline{\kappa}\E$.
\end{theorem}

\begin{remark}
  It will come from the proof (see Theorem \ref{kamtheo}) that the measure of the set $\mathcal{K}_\E$ tends to the measure of the whole phase space as $\E \to 0$.
\end{remark}
\begin{remark}
The proof of the theorem will be constructive, in the sense that one can find numerical values for the relevant constants in the statement, in particular $\underline{\kappa}$. Although possible, this computation may become technically involved and goes beyond the purpose of the present paper. However, we will give some advice on how to get the numerical values.    
  \end{remark}

\begin{remark}\label{remtilt}
The result tells us that in the set $\mathcal{K}_\E$ the confidence ellipse for a nominal solution $(x_0,y_0)$ has the same size as in the linear case. If $\cE_u=0$, then the ellipse is horizontal and $\sigma_x,\sigma_y$ behaves asymptotically for large $N$ as in the linear case. If $\cE_u\neq 0$, then the ellipse is tilted (note that $u_+$ and $u_-$ are orthogonal). Hence, the semi-major axis projects on both coordinate axes. It follows that both uncertainties $\sigma_x,\sigma_y$ are of order $1/\sqrt{N}$. We recover the numerical results in \cite{ssm,smil} in which the same trend appeared for $S$ being the Chirikov standard map and the observations coming from an invariant curve. More details will be given in Section \ref{sec:comparison}.
  \end{remark}

\section{Outline of the proof of Theorem \ref{main}}\label{outline}
The proof of our main theorem is somehow technical. We resume in this section the main ideas, leaving the formal and detailed proof in the next section.

If $\E=0$, map \eqref{kammap} is the integrable twist map
\begin{equation}
  \label{inttwist}
\left\{
\begin{split}
  x_1 &= x+y \\
  y_1 &= y,
\end{split}
\right.
\end{equation}
whose phase space is foliated by invariant curves of the form $\{y=\omega \}$ with $\omega\in (a,b)$. Fixing $\omega$, each orbit on the corresponding curve is given by $(x+n\omega,\omega)_{n\in\ZZ}$, and corresponds to a rotation of angle $\omega$.

If $\E>0$ and small enough, map \eqref{kammap} is a perturbation of map \eqref{inttwist} and KAM theory shows that the phase space is no more foliated by horizontal invariant curves but many of them are preserved although slightly perturbed. More precisely, we introduce the class of Diophantine numbers, given $\gamma>0$ and $\tau>2$,
\begin{equation}
\mathcal{D}_{\gamma,\tau} =\left \{ \omega\in (a,b) \: :\: \left| \omega -\frac{p}{q} \right| \geq \frac{\gamma}{q^\tau} \qquad \forall p,q\in\mathbb{Z}, q\neq 0  \right\}.
\end{equation}
If $\omega\in\mathcal{D}_{\gamma,\tau}$ and $\E$ is small, then the corresponding curve of map \eqref{inttwist} is slightly perturbed becoming an invariant curve of map \eqref{kammap}. This curve can be parametrized by $K_\omega(s) = (1+\psi_\omega(s),\eta_\omega(s))$, $s\in\TT$, and the dynamics on it is conjugated to a rotation of angle $\omega$. With some abuse of notation, we also denote by $K_\omega$ the curve given by the image of $K_\omega(s)$ in the phase space. 
The curve $K_\omega$ is almost horizontal in the sense that
  \begin{equation}\label{estcu}
\norm{\psi_\omega}, \norm{\eta_\omega-\omega} = O(\E).
\end{equation}
It is crucial to remember that the smallness condition on $\E$ is given only in terms of the parameters $\gamma,\tau$ defining the set $\mathcal{D}_{\gamma,\tau}$. The same occurs for the remainder $O(\E)$.
Hence, once the parameters $\gamma,\tau$ are fixed, one gets a large quantity of invariant curves, one for each $\omega\in\mathcal{D}_{\gamma,\tau}$. It turns out that the union of these invariant curves forms a subset $\mathcal{K}_\E$ of the phase space of positive measure that tends to the measure of the phase space as $\E\to 0$. Moreover, all the curves in $\mathcal{K}_\E$ satisfy estimate \eqref{estcu} that is uniform in $\mathcal{K}_\E$ in the sense that the remainder $O(\E)$ does not depend on $\omega$ but only on $\gamma,\tau$. We refer to Theorem \ref{kamtheo} for a detailed statement.

The aim of the proof is to show that the set $\mathcal{K}_\E$ is the one satisfying Theorem \ref{main}. Fix a point $(x,y)\in\mathcal{K}_\E$, consider the curve $K(s)=(s+\psi(s), \eta(s))$ such that $(x,y) = K(s)$ for some $s\in\TT$ and let $\omega$ be its rotation number. It is enough to compute the eigenvalues and eigenvectors of the matrix $\Gamma_N^{-1}= C_N=C_N(x,y) = C_N(K(s))$ defined in \eqref{normmat}.

By the definition of $C_N$, the first step is to find an expression for the matrix $DS(K(s))$ representing the linearized dynamics around an invariant curve in $\mathcal{K}_\E$. Since the dynamics on the invariant curve is conjugated to a Diophantine rotation it is possible to show (see Lemmas \ref{russ_lemma} and \ref{llave_lemma}) that the linearized dynamics is upper triangular in the sense that there exists a matrix $M(s)$ with $\det M(s)=1$ and a constant $\bar{T}$ such that
  \[
M^{-1}(s+\omega) DS(K(s)) M(s) = \left(
\begin{array}{cc}
    1   &  \bar{T}    \\
    0   &  1  
    \end{array}
\right).
  \]
  In our case, we will have $\bar{T}=1+O(\E)$ and
  \[
  M(s) = \left(
\begin{array}{cc}
    1+O(\E)   & O(\E)     \\
    O(\E)   &  1+O(\E)   
    \end{array}
\right).
\]
As before, the remainders $O(\E)$ do not depend on the selected curve but only on $\gamma,\tau$.

Hence, by the chain rule
\begin{equation*}
\begin{aligned}
  C_N(K(s)) &= \sum_{|n|\leq N}[DS^n(K(s))]^TDS^n(K(s)) \\
  &= [M(s)]^{-T}\widetilde{C}_N(s)
  [M(s)]^{-1},
\end{aligned}
\end{equation*}
where
\[
\widetilde{C}_N(s) = \sum_{|n|\leq N}   \left(
  \begin{array}{cc}
    1   &  0  \\
    n\bar{T}  &   1 
    \end{array}
  \right) M^T(s+n\omega)
  M(s+n\omega)  \left(
  \begin{array}{cc}
    1   &  n\bar{T}  \\
    0  &   1 
    \end{array}
  \right). 
\]
A technical computation (see Lemma \ref{lemc}) shows that
 \[
  \widetilde{C}_N =
\left(
\begin{array}{cc}
 (2N+1)(1 + O(\E))   & O(N^2\E)   \\
  O(N^2\E)   &    \left(2N+1+\sum_{|n|\leq N}n^2\right)(1+O(\E))
    \end{array}
\right),
\]
the remainders $O(\E)$ depending only on $\gamma,\tau$.

Since $\det M =1$, it comes easily that the determinant
\[
d_N:= \det C_N(K(s))\sim N^4(1 + O(\E))
\]
for $N$ large. Moreover, it is easy to see that the trace
\[
t_N:= \Tr C_N(K(s))\sim N^3(1 + O(\E)).
\]
 The eigenvalues $\lambda^C_\pm$ of $C_N$ can be computed as roots of the characteristic polynomial giving
\[
\lambda^C_\pm= \frac{1}{2} \left(t_N \pm \sqrt{t_N^2 -4d_N}\right).
\]
It is straightforward that $\lambda^C_+ \sim N^3(1 + O(\E))$, while a more delicate analysis of the expression of $d_N$ and $t_N$ gives $\lambda^C_- \sim N(1 + O(\E))$. We get the eigenvalues $\lambda_\pm$ of $\Gamma_N$ remembering that $\lambda_\pm=1/\lambda^C_\mp$. The eigenvectors come from a direct computation, using carefully the expression of $C_N(K(s))$.

Finally, also the remainders appearing in these last steps only depend on the constants $\gamma,\tau$ defining the set $\mathcal{D}_{\gamma,\tau}$. Hence, the estimates do not depend directly on the selected curve $K(s)$ and can be made uniform in $\mathcal{K}_\E$.

\section{Rigorous proof of Theorem \ref{main}} \label{rig_proof}
\subsection{Notations}\label{secnot}

The complex strip around $\TT$ of width $\rho>0$ is defined as
\[
\TT_\rho = \{x \in \CC/ \ZZ\,:\,| {\Im}\,x|<\rho\}\,,
\]
and we denote by $\bar \TT_\rho$ its closure and by $\partial \TT_\rho = \{|\Im\,x|=\rho\}$ its boundary. 
The set $\Anal(\TT_\rho)$ represents the Banach space of holomorphic functions
$\phi:\TT_\rho \rightarrow \CC$, that can be continuously extended to
$\bar \TT_\rho$, and such that $\phi(\TT) \subset \RR$ (i.e. real-analytic),
endowed with the norm
\[
\norm{\phi}_{\rho} = \sup_{|\Im\, x|\leq \rho} \abs{\phi(x)}\,.
\]

We also introduce a complex strip of $ \TT\times (a,b)\subset \mathbb{A}$ as a complex connected open neighborhood $D\subset \CC/\ZZ\times \CC$ such that $D\cap \mathbb{A} = \TT\times (a,b)$.
We denote by $\Anal(D)$ the Banach space of holomorphic functions
$\varphi:D \rightarrow \CC$, that can be continuously extended to
$\bar{D}$, and such that $\varphi(\TT\times (a,b)) \subset \RR$ (i.e. real-analytic), endowed with the norm
\[
\norm{\varphi}_{D} = \sup_{z \in D} \abs{\varphi(z)}\,.
\]
These notations are extended to vector and matrix valued functions. Finally, the notation $\Emb(\TT_\rho,D)$ will denote the space of holomorphic embeddings $K:\TT_\rho\rightarrow D$ such that each component of $K(s)-(s,0)$ belongs to $\Anal(\TT_\rho)$.  \\
With these notations, the definition of the parameter $\E$ entering in the statement of Theorem \ref{main} is made precise. Since $S$ is real analytic, we can consider its analytic extension to the domain $D=\{(x,y)\in\TT_\rho\times\mathbb{C} \: : \:  \dist(y,[a,b])<\rho  \}$ for some $\rho>0$. We denote, recalling the expression of $S$ in \eqref{kammap},
\begin{equation}
  \label{defe}
\E = \norm{f}_D+\norm{g}_D.
\end{equation}

For $\gamma>0$ and $\tau>2$, the set of Diophantine numbers is defined as
\begin{equation}
  \label{dc}
\mathcal{D}_{\gamma,\tau} =\left \{ \omega\in (a,b) \: :\: \left| \omega -\frac{p}{q} \right| \geq \frac{\gamma}{q^\tau} \qquad \forall p,q\in\mathbb{Z}, q\neq 0  \right\}.
\end{equation}
  
The following notation will be used in the second part of the proof in which we will need several estimates. We denote by $\cE_i$, $i=1,\dots, 26$ any function in $\Anal(\TT_{\rho/4})$ such that $\norm{\cE_i}_{\rho/4}\leq c_i\E$ for a constant $c_i$ only depending on $\rho,\gamma,\tau$. Analogously, given $\alpha\in\mathbb{R}$, we denote by $\cN_j^\alpha$, $j=1,\dots, 7$  any function in $\Anal(\TT_{\rho/4})$ for which there exists a positive constant $\tilde{c}_j$, only depending on $\rho,\gamma,\tau$, such that $\norm{\cN_j^\alpha}_{\rho/4}\leq \tilde{c}_jN^\alpha$ for all integers $N>\frac{1}{\tilde{c}_j}$.
\begin{remark}
The notations $\cE_i$ and $\cN_j^\alpha$ are used to represent formally some remainders of order $\E$ and $N^\alpha$ that can be estimated using only the constants $\rho,\gamma,\tau$. Note that these remainders do not depend on the particular $\omega\in\mathcal{D}_{\gamma,\tau}$. 
\end{remark}

\subsection{The dynamics of the map $S$ and the set $\mathcal{K}_\E$}\label{sec:KAM}

The set $\mathcal{K}_\E$ in Theorem \ref{main} will come from the union of KAM invariant curves. The existence of many invariant curves comes from the classical version of KAM theory, while the fact that invariant curves fill a set of positive measure was proved in various contexts in \cite{alb,arn,laz,pos1}. See also the survey \cite{pos2} on this topic. We refer to the version for mappings given in \cite{shang} which is stated, in our notations, below.

\begin{theorem}
  \label{kamtheo}
  Let $S$ be a real analytic exact symplectic diffeomorphism of the form \eqref{kammap}.
  There exist two constants $\delta,c$ depending only on $\rho,\tau$ and not on $\gamma$ such that if
  \[
\E < \delta \gamma^2
  \]
  then, for every $\omega\in\mathcal{D}_{\gamma,\tau}$ there exists  $K_\omega\in\Emb(\TT_{\rho/2}, D)$ of the form $K_\omega(s)= (s+\psi_\omega(s),\eta_\omega(s))$ such that
  \begin{align}
    \label{inv_eq}
   &S( K_{\omega}(s)) = K_\omega(s+\omega) \quad \mbox{for every }s\in\TT_{\rho/2},
  \end{align}
  and
  \[
\norm{\psi_\omega}_{\rho/2}, \norm{\eta_\omega-\omega}_{\rho/2} < c\gamma^{-2}\E.
\]
  Moreover, the set
  \[
\mathcal{K}_{\E} = \bigcup_{\omega\in\mathcal{D}_{\gamma,\tau}} \left\{K_\omega(s) \: :\: s\in\TT \right\} 
  \]
  has Lebesgue measure
  \[
  \mu(\mathcal{K}_{\E}) \geq (b-a)(1-K_2\gamma) 
  \]
  for a positive constant $K_2$ depending on $\rho,\tau$. 
\end{theorem}

\begin{remark}
Remembering the domain of the diffeomorphism $S$, the factor $(b-a)$ represents the Lebesgue measure of the phase space. Therefore, the invariant curves in $\mathcal{K}_{\E}$ fill the phase space up to a set of Lebesgue measure proportional to $\sqrt{\E}$.
\end{remark}
\begin{remark}
It is possible to go through the proof given in \cite{shang} to get numerical values for the constants involved. More details can be found in \cite{haro_book}, in particular concerning the computation of the parameter $\delta$. The parameter $\gamma$ is free and gives a scale on the 'smallness' of $\E$ and on the measure of $\mathcal{K}_{\E}$. The parameter $\tau$ is also free and enters in the expressions for $\delta,c,K_2$. Finally, the parameter $\rho$ represents the domain of analyticity of the map $S$ and is fixed from the beginning. We will fix the parameters $\gamma<1$ and $\tau>2$ at the beginning of Section \ref{sec:confidecne}.
  \end{remark}

It will be important to study the linearized dynamics around an invariant curve. It turns out that it is upper-triangular. This is sometimes referred to as automatic reducibility (see \cite{haro_llave}). Before considering this result, let us recall the following fundamental result in KAM theory (see, e.g. \cite{russ})

\begin{lemma}\label{russ_lemma}
  Let $\omega\in\mathcal{D}_{\gamma,\tau}$. There exists a constant $c_R$ depending only on $\tau$ such that, for any function $v\in\Anal(\TT_{\rho/2})$ and with zero average, there exists a unique function $u\in\Anal(\TT_{\rho/4})$ with zero average, such that for every $s\in\TT_{\rho/4} $ 
\begin{equation}
  \label{russ}
u(s)-u(s+\omega)=v(s),
\end{equation}
and 
\begin{equation}
  \label{russ_est}
  \norm{u}_{\rho/4}\leq \frac{c_R}{\gamma(\rho/4)^\tau}\norm{v}_{\rho/2}.
\end{equation}
\end{lemma}

\begin{remark}
Estimates on the constant $c_R$ are provided in \cite{russ}. Recent results using computer assisted techniques are presented in \cite{haro_russ}.
  \end{remark}

\begin{lemma}
\label{llave_lemma}
Let $K(s) =(s+\psi(s),\eta(s))^T\in\Emb(\TT_{\rho/2},D)$ represent an invariant curve coming from Theorem \ref{kamtheo} with rotation number $\omega\in\mathcal{D}_{\gamma,\tau}$.\\
Then, there exist a matrix $M\in\Anal(\TT_{\rho/4})$ and a constant $\bar{T} = 1+\cE_1$ such that
  \[
M^{-1}(s+\omega) DS(K(s)) M(s) = \left(
\begin{array}{cc}
    1   &  \bar{T}    \\
    0   &  1  
    \end{array}
\right).
\]
Moreover, $M(s)$ can be written as
\begin{equation*}
 M(s)=  
  \left(
  \begin{array}{cc}
    1+\psi'(s)   &  \frac{-\eta'(s)}{\abs{K'(s)}^2}   \\
     \eta'(s)  &   \frac{1+\psi'(s)}{\abs{K'(s)}^2} 
    \end{array}
  \right)
   \left(
  \begin{array}{cc}
    1   &  u(s)  \\
    0  &   1 
    \end{array}
  \right),
\end{equation*}
with $|u|\leq \cE_2 $.

\end{lemma}

\begin{proof}
 Differentiating \eqref{inv_eq} we get
\begin{equation}\label{inv_diff}
DS (K(s))K'(s) = K'(s+\omega). 
  \end{equation}
Since $K$ is an embedding it is well defined
\[
N(s)= \frac{1}{\abs{K'(s)}^2}\Omega K'(s), \qquad \Omega = \left(
  \begin{array}{cc}
   0   &   -1  \\
    1  &  0  
    \end{array}
  \right),
\]
and we can consider the matrix $M_K\in\Anal(\TT_{\rho/2})$ given by
  \[
M_K(s) = (K'(s) \quad N(s) )=
  \left(
  \begin{array}{cc}
    1+\psi'(s)   &  \frac{-\eta'(s)}{\abs{K'(s)}^2}   \\
     \eta'(s)  &   \frac{1+\psi'(s)}{\abs{K'(s)}^2} 
    \end{array}
  \right).
  \]
  Using \eqref{inv_diff} and the fact that $M_K^{-1}=-\Omega M_K^T\Omega$ we get that
\begin{equation}\label{inv_diff}
M_K(s+\omega)^{-1} DS(K(s))M_K(s) = 
  \left(
  \begin{array}{cc}
    1   &  T(s)   \\
     0  &   1 
    \end{array}
  \right),
  \end{equation}  
where
\begin{equation}\label{formt}
T(s) = N(s+\omega)^T\Omega DS(K(s)) N(s) \in\Anal(\TT_{\rho/2}).
\end{equation}


\noindent Let us denote the average $\bar{T} =\int_{\TT}T(s)ds$. From Lemma \ref{russ_lemma} there exists a unique zero-average function $u\in\Anal(\TT_{\rho/4})$ satisfying \eqref{russ_est} with $v(s) =-T(s)+\bar{T} $ such that, for every $s$,
\begin{equation*}
u(s)-u(s+\omega)=-T(s)+\bar{T}.
\end{equation*}
Consider the matrix
\begin{equation}
  \label{emme}
M(s) = M_K(s)
 \left(
  \begin{array}{cc}
    1   &  u(s)   \\
     0  &   1 
    \end{array}
  \right)\in\Anal(\TT_{\rho/4}) .
\end{equation}
Then, using \eqref{emme} and \eqref{inv_diff},
\begin{equation*}
  \begin{split}
    M(s+\omega)^{-1} DS(K(s))M(s)& =
    \left(
    \begin{array}{cc}
      1   & - u(s+\omega)   \\
      0  &   1 
    \end{array}
    \right)
    \left(
    \begin{array}{cc}
      1   &  T(s)   \\
      0  &   1 
    \end{array}
    \right)
    \left(
    \begin{array}{cc}
      1   &  u(s)   \\
      0  &   1 
    \end{array}
    \right)
     \\
   & =
    \left(
    \begin{array}{cc}
      1   & u(s)- u(s+\omega)+T(s)   \\
      0  &   1 
    \end{array}
    \right)
    =
    \left(
    \begin{array}{cc}
      1   & \bar{T}   \\
      0  &   1 
    \end{array}
    \right).
  \end{split}
\end{equation*}
Finally, from the expression of $T(s)$ in \eqref{formt} and the form of the diffeomorphism $S$ in (\ref{kammap}-\ref{defe}) it holds that
  \[
T(s) = 1+\cE_3, \quad  \bt = 1+\cE_1.
\]
From this, we have $ |T-\bt|\leq\cE_4$ and, from \eqref{russ_est} with $v=T-\bt$,
  \[
|u|\leq \cE_2. 
  \]
The proof of Lemma \ref{llave_lemma} is complete.

  \end{proof}

\begin{remark}
 From the definition of $\cE_i$, the estimates on $\bt$ and $u$ can be written as
    \[
    |\bt -1|\leq c_1\E, \qquad \norm{u}_{\rho/4}\leq c_2\E
    \]
for some constants $c_1$ and $c_2$ only depending on $\rho,\gamma,\tau$. It comes from the proof that to get the numerical values for these constants one needs to compute the constants $c,c_R$ of Theorem \ref{kamtheo} and Lemma \ref{russ_lemma}.  
  \end{remark}

\subsection{The confidence ellipse for nominal values on invariant curves and proof of the main Theorem}\label{sec:confidecne}

We conclude the proof of Theorem \ref{main}. To this aim, we compute the confidence ellipse for nominal solutions on the invariant curves coming from Theorem \ref{kamtheo}.\\

Let $\rho>0$ represent the domain of analyticity of $S$. Let us fix $\tau>2$ and consider the constant $\delta$ coming from Theorem \ref{kamtheo}.
Fix $0<\gamma<1$ such that $\delta\gamma^2<1$ and consider the corresponding Diophantine condition $\mathcal{D}_{\gamma, \tau}$. For any $\E< \delta\gamma^2<1$, we can define the corresponding set $\mathcal{K}_\E$. We do not fix $\E$ and leave it as a free parameter. The aim of the proof to determine a more stringent upper bound $\underline{\kappa}$ for $\E$ depending only on $\rho,\gamma,\tau$. This will be possible thanks to the special form of the estimates in Theorem \ref{kamtheo} and Lemmas \ref{russ_lemma}-\ref{llave_lemma}.   
\begin{remark}\label{rem_ex}
 As an example of how we are going to determine the upper bound for $\E$, let us get an estimate the we will use in the following.
The set $\mathcal{K}_\E$  is made of the union of invariant curves of the form
\[
K(s)=(s+\psi(s), \eta(s)), \quad s\in\TT,
\]
satisfying \eqref{inv_eq} with $\omega\in\mathcal{D}_{\gamma, \tau}$. Moreover, $\psi,\eta\in\Anal(\TT_{\rho/4})$ and satisfy, from Theorem \ref{kamtheo} and Cauchy estimate,
  \[
\norm{\psi'}_{\rho/4}, \norm{\eta'}_{\rho/4} < \frac{4c}{\gamma^2\rho} \E =:c_5\E,
\]
that means $\psi',\eta'  =  \cE_5$.
 Hence, defining $\vv(s): = |K'(s)|^2$, we have
\begin{equation}
  \label{normtg}
\vv(s)= |1+\psi'(s)|^2+|\eta'(s)|^2\leq 1+2(c_5+c_5^2)\E =: 1+c_6\E,
\end{equation}
that means $\vv(s)= 1+ \cE_6$. From this estimate we have that for $\E<\frac{1}{2c_6}$
  \begin{equation}
    \label{tgunif}
\vv = 1+\cE_6 > 1-\norm{\cE_6}_{\rho/4} > 1-c_6\E >\frac{1}{2}
  \end{equation}
 and $c_6$ depends only on $\rho,\gamma,\tau$. Therefore, choosing $\E< \kappa_1 = \min \{\frac{1}{2c_6}, \delta\gamma^2  \}$, we guarantee that estimate \eqref{tgunif} holds for every curve in $\mathcal{K}_\epsilon$. 
\end{remark}

To prove Theorem \ref{main}, fix a point $(x,y)\in\mathcal{K}_\E$ and consider the curve $K(s)=(s+\psi(s), \eta(s)) $ such that $(x,y) = K(s)$ for some $s\in\TT$. We need to compute the eigenvalues and eigenvectors of the matrix $\Gamma_N=\Gamma_N(x,y) = \Gamma_N(K(s))$. From the definition of $\Gamma_N$ it is enough to compute the eigenvalues of the matrix $\Gamma_N^{-1}= C_N=C_N(x,y) = C_N(K(s)) $ defined in \eqref{normmat}.

Since $K$ satisfies \eqref{inv_eq} we can apply Lemma \ref{llave_lemma} and get the existence of a function $u\in\Anal(\TT_{\rho/4})$ with null average and a matrix $M\in\Anal(\TT_{\rho/4})$, of the form
\begin{equation}\label{defm}
 M(s)=  
  \left(
  \begin{array}{cc}
    1+\psi'(s)   &  \frac{-\eta'(s)}{\vv(s)}   \\
     \eta'(s)  &   \frac{1+\psi'(s)}{\vv(s)} 
    \end{array}
  \right)
   \left(
  \begin{array}{cc}
    1   &  u(s)  \\
    0  &   1 
    \end{array}
  \right),
\end{equation}
with $|u|\leq \cE_2$ and such that
\[
DS(K(s)) = M(s+\omega)  \left(
  \begin{array}{cc}
    1   &  \bar{T}  \\
    0  &   1 
    \end{array}
  \right)  M^{-1}(s),
  \]
  for a positive constant $\bar{T}=1+\cE_1$.   By the chain rule, for $n\in\ZZ$,
\[
DS^n(K(s)) = M(s+n\omega)  \left(
  \begin{array}{cc}
    1   &  n\bar{T}  \\
    0  &   1 
    \end{array}
  \right)  M^{-1}(s),
  \]
  that gives
\begin{equation}   \label{formulac}
\begin{aligned}
  C_N(K(s)) &= \sum_{|n|\leq N}[DS^n(K(s))]^TDS^n(K(s)) \\
  &= [M(s)]^{-T}\widetilde{C}_N(s) [M(s)]^{-1},
\end{aligned}
\end{equation}
where
\begin{equation}\label{defctilde}
\widetilde{C}_N(s) = \sum_{|n|\leq N}   \left(
  \begin{array}{cc}
    1   &  0  \\
    n\bar{T}  &   1 
    \end{array}
  \right) M^T(s+n\omega)
  M(s+n\omega)  \left(
  \begin{array}{cc}
    1   &  n\bar{T}  \\
    0  &   1 
    \end{array}
  \right). 
\end{equation}



In the following technical lemma we study the matrix $\widetilde{C}_N$. From now on,  given a function $f(s)$ we denote $f_n=f(s+n\omega)$, $f=f(s)$.

\begin{lemma}
  \label{lemc}
 For every $N>1$, the matrix $\widetilde{C}_N$ is of the form
  \[
  \widetilde{C}_N =
\left(
\begin{array}{cc}
 (2N+1)(1 + \cE_6)   & \cE_7 \cN_1^2   \\
  \cE_7 \cN_1^2   &    \left(2N+1+\sum_{|n|\leq N}n^2\right)(1+\cE_{10})
    \end{array}
\right),
\]
with
    \[
    \det \widetilde{C}_N = \left(2N+1+\mathsmaller{\sum_{|n|\leq N}}n^2\right)(2N+1)(1+\cE_{12}).
    \]
  \end{lemma}
\begin{proof}
 Let us compute, using \eqref{defm},
  \begin{equation*}
  \begin{split}
M^T(s+n\omega) M(s+n\omega) &=  
   \vv_n
  \left(
  \begin{array}{cc}
    1   &  u_n    \\
    u_n   &  u_n^2 + \frac{1}{\vv_n^2} 
    \end{array}
  \right),
  \end{split}
\end{equation*}
and use it into \eqref{defctilde} to get

\begin{equation*}
  \widetilde{C}_N =  
\left(
\begin{array}{cc}
   \sum_{|n|\leq N}\vv_n   & \bt\sum_{|n|\leq N}n\vv_n      \\
   \bt\sum_{|n|\leq N}n\vv_n & \bt^2\sum_{|n|\leq N}n^2\vv_n+\sum_{|n|\leq N} \frac{1}{\vv_n}
    \end{array}
\right) +R_N,
\end{equation*}
where the remainder $R_N$ satisfies, using \eqref{normtg} and the estimates in Lemma \ref{llave_lemma}, 
%
\[
R_N=  
\left(
\begin{array}{cc}
   0   & \sum_{|n|\leq N}\vv_nu_n      \\
 \sum_{|n|\leq N} \vv_nu_n \:  &\:  2\bt\sum_{|n|\leq N} n\vv_n u_n + \sum_{|n|\leq N}\vv_n u^2_n   
    \end{array}
\right)=
\cE_7\left(
\begin{array}{cc}
   0   & \cN_1      \\
 \cN_1 \:  &\:  \cN_2^2   
    \end{array}
\right).
\]
%
From this expressions we get for $N>1$,
\begin{align*}
  \widetilde{C}_N^{(1,1)} &= (2N+1)(1 + \cE_6), \quad  \widetilde{C}_N^{(1,2)}= \widetilde{C}_N^{(2,1)} = \cE_8 \cN_1^2.
\end{align*}
Note that in the second estimate we used that
\[
\bt\sum_{|n|\leq N}n\vv_n=\sum_{|n|\leq N}n(1+\cE_9)= \frac{N(N+1)}{2}\cE_9.
\]
Concerning $\widetilde{C}_N^{(2,2)}$, we use that $\sum_{|n|\leq N} n^2 = \frac{2}{3}N^3+N^2+\frac{N}{3}$ to get
\begin{align*}
  \widetilde{C}_N^{(2,2)} &=
    \left(2N+1+\mathsmaller{\sum_{|n|\leq N}}n^2\right)(1+\cE_8)+\cE_7\cN_2^2 = \left(2N+1+\mathsmaller{\sum_{|n|\leq N}}n^2\right)(1+\cE_{10}).
\end{align*}
Finally, similar computations give
\begin{align*}
  \det\widetilde{C}_N &= \left(2N+1+\mathsmaller{\sum_{|n|\leq N}}n^2\right)(2N+1)(1+\cE_{11})+ (\cE_8\cN_1^2)^2 \\
&= \left(2N+1+\mathsmaller{\sum_{|n|\leq N}}n^2\right)(2N+1)(1+\cE_{12}),
\end{align*}
remembering the definition of $\cN_i^2$.
This concludes the proof of Lemma \ref{lemc}.
\end{proof}

Now we are ready to compute $C_N$ using formula \eqref{formulac}.
%
Lemma \ref{lemc} and \eqref{defm} give
\begin{align}
  \begin{split}
    \label{gammafin}
    & C_N(K(s)) =  M(s)^{-T}\widetilde{C}_N(s) M(s)^{-1} =
    \\
    &=
    I_N
    \left(
    \begin{array}{cc}
      (1+\psi')^2   & (1+\psi')\eta'   \\
      (1+\psi')\eta'   &  (\eta')^2  
    \end{array}
    \right) +
    II_N
    \left(
    \begin{array}{cc}
      -2(1+\psi')\eta'  & (1+\psi')^2-(\eta')^2   \\
      (1+\psi')^2-(\eta')^2   &  2 (1+\psi')\eta' 
    \end{array}
    \right) \\
    &+
        III_N
    \left(
    \begin{array}{cc}
      (\eta')^2   & -(1+\psi')\eta'   \\
      -(1+\psi')\eta'   &  (1+\psi')^2  
    \end{array}
    \right),
  \end{split}
\end{align}
where
 \begin{align*}
 I_N&= (2N+1)(1 + \cE_{13}), \qquad  II_N= \cE_{14}\cN_{3}^2, \\ 
 III_N&= 
 \left(2N+1+\mathsmaller{\sum_{|n|\leq N}}n^2\right)(1+\cE_{15}).
 \end{align*}
 To compute the eigenvalues and eigenvectors we need the trace and the determinant of $C_N(K(s))$.
 From a direct computation we have
   \begin{equation}\label{defdnn}
    d_N := \det C_N(K(s)) = \vv^2(I_NIII_N-II_N^2),
     \end{equation}
but also, since $\det M(s)=1 $, from Lemma \ref{lemc}, 
 \begin{equation*}
   d_N =\det \widetilde{C}_N(s) = \left(2N+1+\mathsmaller{\sum_{|n|\leq N}}n^2\right)(2N+1)(1+\cE_{12}).
 \end{equation*}
%
 Here, a similar argument as in Remark \ref{rem_ex} guarantees the existence of a positive constant $\kappa_2$ only depending on $\rho,\tau,\gamma$ such that if $\E< \kappa_2$  then $d_N> \frac{2}{3}N^4>0$.
Concerning the trace, we have
 \begin{align*}
   t_N:= \Tr C_N(K(s)) = \vv\left( I_N+ III_N  \right) = \left(2(2N+1)+\mathsmaller{\sum_{|n|\leq N}}n^2\right)(1+\cE_{19}).
 \end{align*}
%
 These expressions give us the possibility to compute the eigenvalues $\lambda^C_+=\lambda^C_+(s)$, $\lambda^C_-=\lambda^C_-(s)$ of $C_N(K(s))$ as roots of the characteristic polynomial and get the desired eigenvalues of $\Gamma_N$ as $\lambda_+=\frac{1}{\lambda^C_-(s)}, \lambda_-=\frac{1}{\lambda^C_+(s)}$.\\
To compute $\lambda_+$ we start from the formula

\begin{align*}
  \lambda^C_-&= \frac{1}{2} \left(t_N - \sqrt{t_N^2 -4d_N}\right),
\end{align*}
and note that there exist two positive constants $\kappa_3$ and $\kappa_4$ only depending on $\rho,\tau,\gamma$ such that if $\E< \kappa_3$ and $N>\frac{1}{\kappa_4}$ then $III_N-I_N>\frac{1}{3}N^3$, so that
\begin{align*}
  \sqrt{t_N^2 -4d_N} &= \sqrt{\vv^2(III_N-I_N)^2+4\vv^2II_N^2}\\
  &=\vv(III_N-I_N)+\frac{4\vv^2II_N^2}{\sqrt{\vv^2(III_N-I_N)^2+4\vv^2II_N^2} +(III_N-I_N)} \\
  &= \vv(III_N-I_N) +\cE_{16}\cN_{5}.
\end{align*}
Hence,
\begin{align*}
  \lambda^C_- &= \vv I_N + \cE_{16}\cN_{5} =(2N+1)(1 + \cE_{17}),
\end{align*}
and
\[
\lambda_+ = \frac{1}{2N+1}(1 + \cE_{18}).
\]
To compute $\lambda_-$ we have for $\E< \kappa_3$ and $N>\frac{1}{\kappa_4}$, 
\begin{align*}
\lambda_+^C &=  t_N -
\lambda^C_- =\left(2N+1+\mathsmaller{\sum_{|n|\leq N}}n^2\right)(1+\cE_{19}), 
 %
\end{align*}
so that
\[
\lambda_- = \frac{1}{2N+1+\sum_{|n|\leq N}n^2}(1 + \cE_{20}).
\]

Finally, the eigenvectors $u_+,u_-$ of $\Gamma_N$ corresponding to $\lambda_+,\lambda_-$ satisfy
\[
u_+ = u^C_-, \qquad u_-=u^C_-,
\]
where $u^C_+=u^C_+(s)$,$u^C_-=u^C_-(s)$ are the eigenvectors of $C_N(K(s))$ corresponding to $\lambda^C_+(s),\lambda^C_-(s)$.
Since the matrix is symmetric and positive definite, they are orthogonal and can be written, remembering that $\lambda^C_+ +\lambda^C_- = t_N $, as
\begin{align*}
 u^C_+  &=
\left(
\begin{array}{l}
  C_N^{(2,1)} \\
  \lambda^C_+ - C_N^{(1,1)}
\end{array}
\right) =
\left(
\begin{array}{l}
  C_N^{(2,1)} \\
  t_N - \lambda^C_- - C_N^{(1,1)}
\end{array}
\right)=
\left(
\begin{array}{l}
  C_N^{(2,1)} \\
  C_N^{(2,2)} -\lambda^C_- 
\end{array}
\right).
\end{align*}
Using the definition of $II_N$, we have
\begin{align*}
  C_N^{(2,1)} &= -(1+\psi')\eta'(III_N-I_N)+((1+\psi')^2-(\eta')^2)II_N\\
  &= -(1+\psi')\eta'(III_N-I_N)+\cE_{21}\cN_5^2,
\end{align*}
and
\begin{align*}
  C_N^{(2,2)} -\lambda^C_- &=(1+\psi')^2 III_N +I_N(\eta')^2+2(1+\psi')\eta'II_N-\vv I_N - \cE_{16}\cN_{4} \\
  &= (1+\psi')^2( III_N-I_N) +2(1+\psi')\eta'II_N-\cE_{16}\cN_4 \\
  &= (1+\psi')^2( III_N-I_N) +\cE_{22}\cN_5^2.
\end{align*}

Since $|\psi'|<\cE_1$, we can find a positive constant $\kappa_5<\kappa_3$ only depending $\rho,\tau,\gamma$ such that if $\E< \kappa_5$ and $N>\frac{1}{\kappa_4}$,  
\[
(1+\psi')( III_N-I_N) >\frac{1}{4}N^3,
\]
and we can take
\begin{align}
  \label{vect_fin}
  u_-=
\left(
\begin{array}{l}
 -\eta'+\cE_{23}\cN_{6}^{-1} \\
  (1+\psi')+\cE_{24}\cN_{7}^{-1}
\end{array}
\right)
=
\left(
\begin{array}{l}
 0 \\
  1
\end{array}
\right)
+
\left(
\begin{array}{l}
 \cE_{25} \\
 \cE_{26}
\end{array}
\right)
.
\end{align}
The expression of $u_+$, comes easily from the fact that $u_-\cdot u_+=0$.

We get the thesis choosing $\E < \underline{\kappa}$ and $N>\frac{1}{\underline{\kappa}}$ where $\underline{\kappa} =\min\{\kappa_1,\kappa_2,\kappa_3,\kappa_4,\kappa_5\} $ only depends on $\rho,\tau,\gamma$. Finally, we can take $\overline{\kappa}=\max\{c_{18},c_{20},c_{25},c_{26}\}$ where $c_i$ are the constants appearing in the definition of $\cE_i$.

\section{Comparison with the numerical results}\label{sec:comparison}

The numerical results, well represented in \cite[Fig. 5d]{ssm} and \cite[Fig. 10]{smil}, suggest that, with reference to Figure \ref{figure:ellipse}, both uncertainties $\sigma_x,\sigma_y$ are of order $1/\sqrt{N}$ for large $N$. This is compatible with our results in Theorem \ref{main} if the confidence ellipse is tilted (see Remark \ref{remtilt}). This means that the axes of the confidence ellipse $\mathcal{E}_N(x,y)$, corresponding to the eigenvectors $u_+,u_-$, are not parallel to the Cartesian axes. Actually, in this situation, the larger uncertainty, corresponding to the eigenvalue $\lambda_+$, projects on both coordinate axes.

From the expression of $u_+,u_-$ in Theorem \ref{main} one cannot deduce that the confidence ellipse must be tilted. However, it comes from the proof of the theorem in Section \ref{sec:confidecne} that the eigenvector $u_+$ has a special form. \\
To describe it, we fix $\E<\underline{\kappa}$ and we recall that the set $\mathcal{K}_\E$ is made of invariant curves of the form $K_\omega(s)=(s+\psi_\omega(s),\eta_\omega(s))$ for all $\omega$ satisfying a Diophantine condition. Each nominal solution in $\mathcal{K}_\E$ is of the form $(x,y)=K_{\omega}(s)$ for some $s,\omega$. With this notation, from \eqref{vect_fin}    

\begin{align*}
 u_+
=
\left(
\begin{array}{l}
1+\psi_{\omega}' \\
  \eta_{\omega}'  
\end{array}
\right)+
\cN^{-1}
\left(
\begin{array}{l}
\cE \\
  \cE 
\end{array}
\right),
\end{align*}
where $(1+\psi_{\omega}',\eta_{\omega}')^T$ is the tangent vector to curve $K_{\omega}$ in $(x,y)$. More precisely, the vertical component of $u_+$ can be written as $\eta_{\omega}'+\cN^{-1}\cE$ with
  \[
\cN^{-1}\cE \leq c_{23}\tilde{c}_6N^{-1}\E< c_{23}\tilde{c}_6N^{-1}\underline{\kappa}, \qquad \mbox{for } N>\frac{1}{\tilde{c}_6}.
  \]
  Therefore if  $\eta_{\omega}'\neq 0$ and $N>N_0$ with
  \[
N_0=\max\left\{ \frac{1}{\tilde{c}_6}, \frac{c_{23}\tilde{c}_6\underline{\kappa}}{\eta_{\omega}'}\right \}
  \]
  then the eigenvector $u_+$ is not tangent to the coordinate axes.\\
We will see later that the condition $\eta_{\omega}'\neq 0$ is of full measure in $\mathcal{K}_\E$. However, this result does not give a generalization of Theorem \ref{main} since $N_0$ depends on $\eta_{\omega}'$ and is no more uniform in $\mathcal{K}_\E$. \\
We conclude proving that the condition $\eta_{\omega}'\neq 0$ is of full measure in $\mathcal{K}_\E$. Let $\gamma>0$, $\tau>2$ be as in Section \ref{sec:confidecne} and consider the set $\mathcal{D}_{\gamma,\tau}$ defined in \eqref{dc}.
Fix $\E$ as in Theorem \ref{main} and denote by $\mathcal{K}_0$ the set of points in $\mathcal{K}_{\E}$ such that $\eta_\omega'(s) = 0$ for some $\omega\in\mathcal{D}_{\gamma,\tau}$ and $s\in [0,2\pi]$. The next lemma proves that $\mathcal{K}_0$ has zero measure.
\begin{lemma}\label{lem_zero}
  If the function $g$ in \eqref{kammap} is not identically zero, then the set $\mathcal{K}_0$ defined as before has zero Lebesgue measure.
  \end{lemma}
\begin{proof}
  Since the invariant curves $K_\omega$ are analytic, for every $\omega$ we have that either $\eta'_\omega \equiv 0$ or $\eta_\omega'(s)$ vanishes in a set $N_{\omega}\subset [0,1]$ made of a finite number of points. We first prove that the set
  \[
  \bigcup_{\omega\in\mathcal{D}_{c,\tau}}\left\{(s+\psi_\omega(s),\eta_\omega(s)) \: : \: s\in [0,1], \eta'_\omega\equiv 0  \right\}
  \]
  has zero measure. Actually, by contradiction, using the invariance and the second equation in \eqref{kammap}, the analytic function $g(x,y)$ would vanish on a set of positive measure. Hence it would be identically zero contradicting our hypothesis on $g$.\\
Consider now the set
  \[
  \mathcal{A}= \bigcup_{\omega\in\mathcal{D}_{c,\tau}}\left\{(s+\psi_\omega(s),\eta_\omega(s)) \: : \: s\in N_\omega \right\},
  \]
  recalling that, for each $\omega$, the set $N_\omega\subset [0,1]$ has finite cardinality. We prove that also $\mathcal{A}$ has zero measure. Since $1+\psi'_\omega>0$ we can reparametrize $K_\omega$ as $(t,\gamma_\omega(t))$ and $\gamma'_\omega(t) =0 \Leftrightarrow \eta'_\omega(s)=0$. Hence, denoting by $\chi(\mathcal{A})$ the characteristic function of $\mathcal{A}$
  \begin{align*}
    \mu(\mathcal{A}) &= \int_\RR\left(\int_0^{1} \chi(\mathcal{A})     dt \right)dy \leq \int_{\mathcal{D}_{c,\tau}}\left(\int_0^{1} \chi(\mathcal{A})\sqrt{1+(\gamma'_\omega(t))^2} dt    \right)d\omega \\
    & = \int_{\mathcal{D}_{c,\tau}}\left(\int_{K_\omega} \chi(\mathcal{A})ds    \right)d\omega =0
    \end{align*}
  since, for every $\omega$ the set $\mathcal{A}\cap K_\omega$ is made of a finite number of points. This proves Lemma \ref{lem_zero}. 
  \end{proof}



\section{Conclusions}\label{sec:conclusions}
We considered the problem of orbit determination supposing that the
number of observations grows simultaneously with the time span over
which they are performed. We considered the case of analytic
perturbations of the integrable twist map of the cylinder and gave an
analytical description of the confidence region in a set of positive
measure as the number of observation grows.

This is an analytical proof of some numerical results obtained in
\cite{ssm,smil}. Our result covers the case of estimating only the
initial conditions in regular zones. The numerical results cover a
wider situation including the estimation of dynamical parameter and
chaotic zones. An analytical study of these cases will be the aim of
future works.

The problem can also be generalized to many other interesting settings
such as more degrees of freedom, continuous dynamics and different
observations processes. For example, as suggested in Remark \ref{remextend}, one can set up an orbit determination process in the case of Hamiltonian systems and define the corresponding confidence ellipse and Covariance Matrix. In analogy with our present results, invariant tori should take the place of invariant curves and the used results of KAM theory should be replaced by the corresponding ones for Hamiltonian systems. However, despite this apparent analogy, a generalization of our results to this setting seems not to be straightforward.

We will also stress that our result can be interpreted a step towards
the understanding of the relation between the orbit determination and
the dynamics.

\section*{Acknowledgements}
  This problem was proposed to me by Andrea Milani. This result and possible further developments are dedicated to his memory.

\end{document}